\documentclass[12pt, letterpaper]{article}

\usepackage[english]{babel}
\usepackage{hyperref}
\usepackage[T1]{fontenc}

\usepackage{amssymb}
\usepackage{amsmath}
\usepackage{amsfonts}
\usepackage{amsthm}
\usepackage{mathrsfs}

\usepackage{stmaryrd}
\usepackage{epsfig,psfrag,graphicx,epstopdf}
\usepackage[usenames,dvipsnames]{color}
\usepackage[top=2cm, bottom=2cm, left=2cm, right=2cm]{geometry}
\usepackage{color}

\swapnumbers

\newtheorem{thm}{Theorem}[section]
\newtheorem{cor}[thm]{Corollary}
\newtheorem{lem}[thm]{Lemma}
\newtheorem{prop}[thm]{Proposition}

\theoremstyle{definition}
\newtheorem{defin}[thm]{Definition}

\newtheorem*{xrem}{Remark}
\newtheorem*{xrems}{Remarks}

\newenvironment{dedication}
    {\vspace{6ex}\begin{quotation}\begin{center}\begin{em}}
    {\par\end{em}\end{center}\end{quotation}}

\newcommand{\R}{{\mathord{\mathbb R}}}

\newcommand{\V}{\ensuremath{\mathbf{v}}}
\newcommand{\el}{\ensuremath{\mathscr{L}}}
\newcommand{\eg}{\ensuremath{\mathscr{G}}}

\newcommand{\cX}{{\cal X}}
\renewcommand{\a}{\alpha}
\newcommand{\ba}{\underline{\alpha}}

\def\<#1>{\langle#1\rangle}

\def\Xint#1{\mathchoice
{\XXint\displaystyle\textstyle{#1}}%
{\XXint\textstyle\scriptstyle{#1}}%
{\XXint\scriptstyle\scriptscriptstyle{#1}}%
{\XXint\scriptscriptstyle\scriptscriptstyle{#1}}%
\!\int}
\def\XXint#1#2#3{{\setbox0=\hbox{$#1{#2#3}{\int}$ }
\vcenter{\hbox{$#2#3$ }}\kern-.6\wd0}}

\def\dashint{\Xint-}

\numberwithin{equation}{section}
\begin{document}

\makeatletter{\renewcommand*{\@makefnmark}{}
\footnotetext{Work partially supported by U.S. National Science Foundation
grant DMS 1301555\\
\copyright\, 2013 by the authors. This paper may be reproduced, in its 
entirety, for non-commercial purposes.}\makeatother}

\title{The Kac Model Coupled to a Thermostat}

\author{\vspace{5pt} Federico Bonetto$^{1}$, Michael Loss$^{1}$, and Ranjini 
Vaidyanathan$^{1}$ \\
\vspace{5pt}\small{$1.$ School of Mathematics, Georgia Institute of 
Technology}\\[-6pt]
\small{
686 Cherry Street
Atlanta, GA 30332-0160 USA}\\
 }

\maketitle

\begin{dedication}
We dedicate this article to our friend, teacher, and co-author Herbert Spohn.
\end{dedication}

\begin{abstract}
In this paper we study a model of randomly colliding particles interacting with 
a thermal bath. Collisions between particles are modeled via the Kac master 
equation while the thermostat is seen as an infinite gas at thermal equilibrium 
at inverse temperature $\beta$. The system admits the canonical distribution at 
inverse temperature $\beta$ as the unique equilibrium state. We prove that any 
initial distribution approaches the equilibrium distribution exponentially fast 
both by computing the gap of the generator of the evolution, in a proper 
function space, as well as by proving exponential decay in relative entropy. We 
also show that the evolution propagates chaos and that the one particle 
marginal, in the large system limit, satisfies an effective Boltzmann-type 
equation.

\end{abstract}

\medskip
\leftline{\footnotesize{\qquad Mathematics subject classification numbers:  
47A63, 15A90}}
\leftline{\footnotesize{\qquad Key Words: spectral gap, kinetic theory}}

\section{Introduction} \label{intro}

The master equation approach to kinetic theory has had a revival in recent 
years. It was introduced by Mark Kac in 1956 \cite{kac} to model a system 
of $N$ particles interacting through a Markov process. In its basic form, after 
waiting an exponentially distributed time,  one selects randomly and uniformly a 
pair of particles and lets them collide with a random scattering angle.  One 
assumes a spatially homogeneous situation in which the state of the system is 
entirely specified by the velocities of the particles. The time evolution for 
the probability distribution of finding the system in a given state is then a 
{\it linear} master equation (albeit in very high dimensions) called the Kac 
master equation.

The model is based on clear probabilistic assumptions, and its simplicity 
allows one to focus on central issues that are very difficult to study in more 
fundamental models like Newtonian mechanics. Kac's main motivation was to give a 
rigorous derivation of the {\it non-linear} spatially homogeneous Boltzmann 
equation (\cite{kac}, see also \cite{McK}). This is based on the notion of 
chaotic sequences (which Kac called `sequences that have the Boltzmann 
property'). A derivation of the Boltzmann equation from the laws of classical 
mechanics is much more difficult and so far has only been achieved for 
situations with very few collisions 
\cite{lanfordproof,lanfordreview,illnerpulvirenti}.

The Kac master equation also yields some insight into the the central question 
of approach to equilibrium for large particle systems. Kac suggested using 
the gap as a measure for the rate of approach to equlibrium and he conjectured 
that it is bounded below by a positive constant independent of the particle 
number $N$ \cite{kac}. This conjecture was first proved in \cite{jeanvresse} 
and shortly  thereafter the gap was computed exactly \cite{CCL,Maslen}.  

It turns out that while the gap is a good notion for measuring the rate of 
approach once the system is close to equilibrium, this is not the case far away 
from equilibrium. The gap measures the rate at which the $L^2$-norm of the 
deviation from equilibrium tends to zero. Since probability distributions 
in $N$ variables that are close to a product have an $L^2$-norm which is of the 
order of $C^N,\, C>1$, one has to wait times of order $N$ until this norm falls 
below a fixed number. This is clearly not physical. Furthermore, in all 
the exact calculations, the gap as well as the single particle marginal of the 
gap eigenfunction approach the corresponding quantities of the {\it linearized} 
Boltzmann equation, as $N \to \infty$.

A better notion of equilibration is in the entropic sense. It is easy to 
prove that the relative entropy of any state (c.f. eq.~\eqref{E:relent}) is 
decreasing in time (note that we define entropy with opposite sign). One would 
expect the {\it entropy production}, i.e. the negative  time derivative of the 
entropy, to be proportional to the entropy itself because one expects 
exponential decay of the entropy. The results in this direction have been 
disappointing. It turns out that there are states whose entropy production is 
inversely proportional to the particle number $N$. In \cite{villani} a lower 
bound inversely proportional to $N$ was given while an upper bound of the same 
order was proven in \cite{amit}. The upper bound was achieved by estimating the 
entropy production of a state in which approximately half of the kinetic 
energy is stored in $N\delta_N$ particles and, of course, the remaining 
kinetic energy is stored in $N(1-\delta_N)$ particles. It was shown in 
\cite{amit} that the ratio of entropy production and the entropy is bounded 
above by  $C_\beta N^{-1+\beta}$ for any $\beta >0$ provided one chooses 
$\delta_N$ as a suitable inverse power of $N$. Clearly, such a state, in 
which a few molecules contain half of the total kinetic energy is not observed 
in nature. This raises the question of how to characterize those states for 
which the entropy converges to zero on a reasonable time scale. Our very 
preliminary answer is to consider  states in which an `overwhelming' number of 
particles with kinetic energy per particle of the order $\beta^{-1}$ are in 
equilibrium and only a few particles, a `local'  disturbance are out of 
equilibrium.

How should one describe such a state in the context of a model that is spatially 
homogeneous in the first place? In this paper we address this question by 
coupling a system of particles, `the small system out of equilibrium' to a heat 
bath, `the large system in equilibrium', i.e., we couple the Kac model to a 
thermostat. This idea is not new. There has been work in \cite{Frohlich-Student} 
on the (spatially inhomogeneous) Boltzmann equation coupled to a thermostat 
where approach to equilibrium was proved. Likewise there has been recent 
work in \cite{BCKLs,BCELM} that considered particles in an electric field 
interacting with external scatterers (billiard), where the 
thermostat is given by a deterministic friction term while the collision with 
the obstacles provide stochasticity. 

In the present work we focus exclusively on the approach to equilibrium of a 
system of particles in one space dimension having unit mass interacting with 
a thermostat in the context of the Kac master equation. There are a number of 
ways to model this and the route taken here is to describe the thermostat as 
the interaction of particles of the Kac system with particles that are already 
at equilibrium, i.e., whose distribution is given by a Gaussian with inverse 
temperature $\beta$. Moreover we assume that the thermostat is much larger
than the system so that every particle in the thermostat collides at most once 
with a particle in the Kac system.

Calling $f_t(\V)$, $\V\in\R^N$, the probability distribution of 
finding the system with velocities $\V$ at time $t$ , our master equation is given by
\begin{equation} \label{E:mast}
\frac{\partial f}{\partial t} = -\eg f:=  -\lambda N(I-Q)[f] - \mu 
\sum_{j=1}^{N}{(I-R_j)}[f] \ .
\end{equation}
The first term $\eg_K := N(I-Q)$ describes the collision among the particles and has the usual 
Kac form
\[ 
Q[f](\V) := \frac{1}{\binom{N}{2}} \sum_{i<j} 
\dashint_{0}^{2\pi} f(\V_{i,j}(\theta))d\theta 
\]
with
\begin{align*}
&\V_{i,j}(\theta)=(v_1,\ldots,v_i^*(\theta),\ldots,v_j^*(\theta),
\ldots , v_N)\crcr
&v_i^*(\theta) = v_i \cos{(\theta)} + v_j 
\sin{(\theta)}\qquad \,\qquad v_j^*(\theta) 
= -v_i \sin{(\theta)} + v_j \cos{(\theta)} 
\end{align*}
while the second term $\eg_T := \sum_{j=1}^{N}{(I-R_j)}$ describes the interaction with the thermostat where
\[
R_j[f]:= \int{dw \dashint_{0}^{2\pi}{d\theta \sqrt{\frac \beta {2\pi}} e^{-\frac \beta 2 
w_j^{*2}(\theta)} f(\mathbf{v}_j(\theta,w))}} 
\]
and $\mathbf{v}_j(\theta,w)=(v_1,..., v_j \cos{(\theta)} + w 
\sin{(\theta)},...,v_N) $, $w_j^*(\theta) = -v_j \sin{(\theta)} + w 
\cos{(\theta)}$. We use the notation
\[
 \dashint_a^b f(\theta)d\theta=\frac{1}{b-a}\int_a^b f(\theta)d\theta.
\]

In contrast to Kac's original work, where the configuration space was the 
constant-energy sphere, here the configuration space is $\mathbb{R}^N$ since the 
energy is not conserved.

As we will see, the time evolution defined by~\eqref{E:mast} is 
ergodic and has the unique equilibrium state 
\[
\gamma(\V):= \prod_{j}{g(v_j)} \ ,
\] 
where $g(v)= \sqrt{\frac \beta {2\pi}} e^{-\frac \beta 2 v^2}$. 

As mentioned before, the approach to equilibrium can be measured in a 
quantitative fashion by the gap of the operator $\mathcal G$. While not 
self-adjoint on the space $L^2(\R^N, d\V)$, a ground state transformation can be 
performed that brings this operator into a self-adjoint form on the space 
$L^2(\R^N, \gamma(\V)d\V)$. Writing
\[
 f(\V)=\gamma(\V)(1+h(\V))
\]
in eq.~\eqref{E:mast} we get the new equation
\begin{equation} \label{E:hmast}
 \frac{\partial h}{\partial t} = -\mathcal{L} f:=-\lambda N (I-Q)[h]- \mu 
\displaystyle\sum_{j=1}^{N}{(I-T_j)[h]}
\end{equation}
where 
\[
T_j[h]=\int dw g(w) \dashint d\theta 
h(\V_j(\theta,w)).
\]
We consider $h$ in the space $L^2(\R^N, \gamma(\V)d\V)$ with inner product 
$\langle h_1,h_2 \rangle:=\int{h_1 h_2 \gamma d\V}$. Defining the spectral gap 
as \[
\Delta_N := \inf \{ |\langle h, \el h \rangle| : ||h||=1, \langle h,1 \rangle=0 \}
\]
we get the following:

\begin{prop}\label{P:1gap} We have
\[
\Delta_N = \frac \mu 2. 
\] 
The corresponding eigenfunction is 
\[
h_{\Delta_N}(\V):=\sum_{i=1}^N\left(v_i^2 -\frac{1}{\beta}\right).
\]
\end{prop}
Thus, the spectral gap does not depend on the parameter $\lambda$ of the Kac 
operator. To have a more precise idea of the role of $\lambda$ in the 
equilibration process we also study the ``second'' spectral gap defined as
\[
\Delta_N^{(2)} := \inf \{ |\langle h, \el h \rangle| : ||h||=1, \langle h,1 \rangle=0, 
\langle h,h_{\Delta_N} \rangle=0 \} \ .
\]

We have the following theorem that renders the gap $\Delta_N^{(2)}$  in an 
explicit fashion. 

\begin{thm} \label{T:2gap}
$\Delta_N^{(2)}$ is given by the lower root $a_2$ of the quadratic equation
 \begin{equation} \label{E:2gap}
  x^2 - \left(\lambda \Lambda_N + \frac{13}{8} \mu\right)x + \mu 
\left(\lambda \Lambda_N + \frac{5}{8}\mu\right) -
\frac{3}{8} \lambda \Lambda_N \mu \left(\frac {3}{N+2}\right) =0 \ ,
\end{equation}
where $\Lambda_N = \frac{1}{2}\frac{N+2}{N-1}$. The corresponding eigenfunction 
is an even polynomial of degree 4 in all the $v_i$.
\end{thm}

As $N \to \infty$ one finds for the gap 
\begin{equation} \label{E:gapinf}
\Delta^{(2)}_\infty =  \min\left\{ \frac{\lambda}{2} + \frac{5}{8} 
\mu,\mu\right\} \ .
\end{equation}

As explained before, a deeper way of understanding the approach to equilibrium 
is to consider the entropy production. In the model at hand one is indeed in the 
lucky situation that the interaction with the thermostat yields a decay rate for 
the entropy, uniformly in $N$. The relative entropy of a state $f$ with respect 
to the equilibrium state $\gamma$ is defined as 
\begin{equation} \label{E:relent}
S(f|\gamma) := \int_{\R^N} f(\V) \log \frac{f(\V)}{\gamma(\V)} d \V \ .
\end{equation}

\begin{thm}\label{T:conen}
Let $f_t$ be the solution of the master equation~\eqref{E:mast} with initial 
condition $f_0$. Then
\[
S(f_t|\gamma) \le e^{-\rho t} S(f_0|\gamma) \ ,
\]
where
\[
 \rho=\frac{\mu}{2} \ .
\]
\end{thm}

The central concept for the derivation of the Boltzmann equation is that of a
chaotic sequence. More 
precisely, given a distribution $f^{(N)}(\V)$ with $\V\in\R^N$, we can 
define the $k$ particle marginal as
\[
 f_k^{(N)}(v_1,\ldots,v_k)=\int f^{(N)}(\V)\prod_{i=k+1}^N dv_i \ .
\]

\begin{defin} 
A sequence of probability distributions $\{ f^{(N)}(\mathbf{v}) 
\}_{N=1}^\infty$ on $\R^N$ is said to be \textbf{chaotic} if, $\forall 
k \geq 1$, we have
\begin{equation*}
\lim_{N \to \infty} {f_k^{(N)}(v_1,...,v_k)} = \lim_{N \to 
\infty} {\prod_{j=1}^k {f_1^{(N)}(v_j)}} \ ,
\end{equation*}
where the above limit is taken in the weak sense.
\end{defin}
Consider now a chaotic family of initial conditions $f^{(N)}(\V)$ for eq.~\eqref{E:mast}. With a simple generalization of the proof by Kac \cite{kac} 
(see also \cite{McK}), we can show that the solution at time $t$, 
$f_t^{(N)}(\V)$, is also chaotic. This property is called propagation of 
chaos. It follows that if we define
\[
 \overline{f}_t(v_1)=\lim_{N\to\infty}\int f_t^{(N)}(\V)\,dv_2\cdots dv_N \ ,
\]
eq.~\eqref{E:mast} gives rise to the effective evolution for $\overline{f}_t$, 
that is

\begin{thm}\label{T:Bolt}
$\overline{f}_t(v)$ is the solution of the following `Boltzmann Equation': 
\begin{align*}
\frac{\partial \overline{f}_t(v)}{\partial t} &= 2 \lambda \dashint{ d\theta 
\int{dw[\overline{f}_t(v \cos{\theta}+w \sin{\theta})\overline{f}_t(-v 
\sin{\theta}+w \cos{\theta})-\overline{f}_t(v)\overline{f}_t(w)]} }\\
&+ \mu [ \int{dw \dashint{d \theta g(-v \sin{\theta}+w \cos{\theta}) 
\overline{f}_t(v \cos{\theta}+w \sin{\theta})}}-\overline{f}_t(v)] 
\end{align*}
with $\overline{f}_0(v)$ as initial condition.
\end{thm}

Our proof of propogation of chaos, following the argument in \cite{kac}, does not establish the validity of the above equation \emph{uniformly} in $t$. One can hope to achieve this by adapting, to our model, the argument in \cite{MM}, where propagation of chaos uniform in $t$  is shown for the Kac model.

One can linearize the above Boltzmann equation about the ground state and study 
the operator associated with this evolution. It turns out that the Hermite 
polynomials diagonalize both the collision part and the thermostat, with the 
$n$-th degree polynomial $H_n(v)$ yielding eigenvalues $2\lambda(1-2s_{n})$ and 
$\mu(1-s_n)$, respectively, where 
$s_n:=\displaystyle\dashint_0^{2\pi}{\cos^n{\theta}d\theta}$. Thus, the gap is 
$\frac{\mu}{2}$, and the ``second'' gap is $\frac{\lambda}{2}+\frac{5}{8}\mu$, 
which correspond to the $N \to \infty$ limit of the respective gaps found at the 
Master equation level (Theorem~\ref{T:2gap}). Incidentally, the eigenvalue 
$\mu$ found in the latter (see eq.~\eqref{E:gapinf}) does not appear here since 
the single-particle marginal of the corresponding eigenfunction vanishes in the 
limit.
\begin{xrem}
As we will see, the proof of Theorem~\ref{T:conen} (in particular, Proposition~\ref{P:N1}) also proves that the relative entropy associated with the Boltzmann 
equation decays at the same rate $\rho$. \end{xrem}

It is interesting to note that if we define the total kinetic energy $K(f)$ as 
\[
K(f) := \frac12\sum_{i=1}^N \int_{\R^N} v_i^2 f(v_1, \dots, v_N) dv_1 \cdots 
dv_N \ ,
\]
we get, using eq.~\eqref{E:mast}, that

\begin{align*}
\frac{dK}{dt} &= - \mu N K + \frac{\mu}{2} \sum_j \int dw d\V \dashint d\theta 
\biggl(\sum_k v_k^2\biggr)g(w_j^*(\theta))f(\V_{j}(\theta,w)) \crcr
&= - \mu N K + \frac{\mu}{2} \sum_j\biggl(\sum_{k \neq j}\int d\V v_k^2 f(\V) + 
\int dw d\V\,g(w)f(\V) \dashint d\theta (v_j \cos(\theta)+w\sin(\theta))^2 \biggr) \ ,
\end{align*}
where the Kac collision gives no contribution as it preserves the total kinetic energy.
This yields

\begin{equation}\label{E:New}
\frac{dK}{dt} = -\frac{\mu}{2}\left(K- \frac{N}{2\beta}\right) \ .
\end{equation}
One can interpret eq.~\eqref{E:New} as Newton's law of cooling.
This law, however,  is usually stated in terms of the 
{\it temperature} of the system at time $t$, i.e.,
\begin{equation}\label{E:temp}
 \frac12 T(t) := \frac{K(t)}{N} \ .
\end{equation}
  It is far from clear that 
eq.~\eqref{E:temp} can be used to define the temperature of a state far from 
equilibrium. To make such an identification one would have to show that the full 
distribution $f_t(\V)$ is close, in a meaningful sense, to a Maxwellian with 
temperature $T(t)$. In general we see no reason why this should be true. We 
believe, however,  that in the case of an `infinitely slow' transformation, i.e. 
the case where $\mu$ is very small relative to $\lambda$, the collisions provide 
enough `mixing' to guide the evolution along Maxwellians.

The plan of the paper is the following. In Section~\ref{sec:2.1} the gap is 
computed and in Section~\ref{sec:2.2} the rate of decay of the relative entropy 
is established. In Section~\ref{sec:3} we show propagation of chaos and 
we finish with a few remarks.

\section{Approach to Equilibrium: Proof of Statements} \label{sec:2}

Before proceeding with the study of the approach to equilibrium, we 
observe that by choosing appropriate units of energy, we can set $\beta=1$ without loss of generality.

\subsection{Approach to Equilibrium in $L^2$} \label{sec:2.1}

In this section we study the lower part of the spectrum of the operator $\el$ 
defined in eq.~\eqref{E:hmast} acting on the Hilbert space 
$\cX=L^2(\R^N,\gamma(\V)d\V)$. To distinguish the action of the 
thermostat from that of the Kac collisions we define the operators
\[ 
\el_T:= \displaystyle\sum_{j=1}^{N}{(I-T_j)} \qquad\qquad \el_K:=  N 
(I-Q) 
\]
so that
\[
\el= \mu\el_T + \lambda\el_K.
\]
It is easy to see that the operator $\el$ for the evolution of $h$ is 
self-adjoint on $\cX$. Moreover $\el$ preserves the subspace of $\cX$ 
formed by the functions symmetric under permutation of the variables. 

To begin, we report some known or simple results on the spectra of $\el_K$ 
and $\el_T$. 
We say that a function $h(\V)$ is {\it radial} if it depends only on 
$r^2=\sum_i v_i^2$. We call $\cX_r$ the subspace of $\cX$ of radial 
functions, and $\cX^\perp$ the subspace of functions 
orthogonal to the constant function, i.e. $\cX^\perp=\{h\in \cX\,|\, 
\langle h,1 \rangle=0\}$. We have

\begin{lem} \label{L:elK} \hspace*{\fill} 
\begin{itemize}
\item[$\bullet$] $\el_K \geq 0$, $\el_T \geq 0$.
\item[$\bullet$] $\el_K[h]=0 \Leftrightarrow h\in\cX_r$, and $\el_T[h]=0 \Leftrightarrow h=\text{constant}$.
\end{itemize}
\end{lem}
  
\begin{proof}
All claims follow from the following observations: 
\[ 
 2 \langle (I-Q)h, h \rangle = 
\frac{1}{\binom{N}{2}}\sum_{i<j}{\dashint{d\theta 
\int_{\R^N}{|h(\mathbf{v}_{i,j}(\theta))-h(\mathbf{v})|^2\gamma 
d\mathbf{v}}}} \geq 0 
\]
\[
2 \langle \sum_j{(I-T_j)}h,h \rangle = \sum_j \bigl( \dashint{d\theta \int d\V 
dw g(w) 
\gamma(\V) |h(\V_j(\theta,w))-h(\V)|^2}\bigr) \ge 0
\ ,\]
the first of which is an identity due to Kac \cite{kac}.
\end{proof}
Notice that the Kac operator alone acting on $\mathbb{R}^N$ has a degenerate 
ground state. 

From the above Lemma, we see that the unique equilibrium state 
corresponding to eq.~\eqref{E:hmast} is $h(\V)=1$.

The following Theorem is a direct consequence of the results in \cite{CCL}.

\begin{thm}[\cite{CCL}] \label{T:CCL}
We have that 
\[
\Lambda_N := \inf \{|\langle h, \el_K h \rangle|: ||h||=1, h \perp \cX_r \} = 
\frac{1}{2} \frac{N+2}{N-1} 
\]
and the corresponding eigenfunction is 
$\sum_{j=1}^{N}{v_j^4}-\frac{3}{N+2}\left( \sum_{j=1}^N v_j^2\right)^2$.
\end{thm}
  
To study the spectrum of $\el_T$ we use the Hermite polynomials $H_\alpha(v)$ 
with weight $g(v)$. More precisely, for $\alpha$ integer, we set 
\[
 H_\a(v)=(-1)^\alpha 
e^{\frac{v^2}{2}}\frac{d^\alpha}{dv^\alpha}e^{-\frac{v^2}{2}}
\]
so that
\begin{enumerate}
 \item $H_\alpha(v)$ is a polynomial of degree $\alpha$. Moreover 
$H_\alpha(-v)=(-1)^{\alpha}H_\alpha(v)$.
\item The coefficient of $v^\alpha$ in $H_\alpha$ is 1.
\item The $H_\alpha$ are orthogonal in $L^2(\R,g dv)$. More precisely
\[
 \int 
H_{\alpha_1}(v)g(v)H_{\alpha_2}(v)dv=\sqrt{2\pi}\alpha_1!\delta_{\alpha_1,
\alpha_2}.
\]
\end{enumerate}

\begin{lem} \label{L:herm}
$H_\alpha(v_j)$ form an orthogonal basis of eigenfunctions for the operator 
$T_j$ and $T_jH_\alpha=s_\alpha H_\alpha$ with $s_\alpha=0$ if $\alpha$ is odd 
while
\[
s_{2\alpha}= \dashint_{0}^{2\pi}{d\theta 
\cos^{2\alpha}{\theta}}=\frac{(2\alpha)!}{2^{2\alpha}\alpha!^2} \ .
\]
\end{lem}
\begin{proof}
We drop the subscript $j$ here for ease of notation. First, we observe that 
\begin{align*}
\int{T[H_\alpha(v)] H_n(v) g(v)dv} &=\int{dw dv g(v) g(w) H_n(v) \dashint{d\theta H_\alpha(v\cos{\theta}+w\sin{\theta})}} \crcr
&= \int{dw dv g(v) g(w) H_\alpha(v) \dashint{d\theta H_n(v\cos{\theta}+w\sin{\theta})}}.
\end{align*}
Since $T[H_\alpha(v)]$ is a polynomial in $v$ of degree $\alpha$, the first line 
implies that $\int{T[H_\alpha(v)] H_n(v) g(v)dv}=0$ if $n>\alpha$. Likewise, the 
second line implies that $\int{T[H_\alpha(v)] H_n(v) g(v)dv}=0$ if $\alpha>n$. 
Thus, \[ T[H_\alpha(v)]=c_\alpha H_\alpha(v) \ .\] By equating the coefficients 
of $v^\alpha$ in the above, we get that 
$c_\alpha=\dashint{\cos^{\alpha}{\theta}}=s_\alpha$.
\end{proof}

Note that $s_{2(\alpha+1)}<s_{2\alpha}$ and $s_{2\alpha}\to 0$ as 
$\alpha \to\infty$. Since $\el_T$ is just the direct sum of $(I-T_j)$ we get

\begin{cor} \label{C:elT}
The functions
\[
 H_{\ba}(\V):=\prod_{i=1}^N H_{\a_i}(v_{i}) \ ,
\]
where $\ba=(\a_1,\ldots,\a_N)$, is an eigenfunction of $\el_T$ with 
eigenvalue
\[
 \sigma_{\ba}:=\sum_i{(1-s_{\a_i})}.
\]
The set $\left\{H_{\ba}\right\}_{\ba\geq 0}$ form an orthogonal 
basis of eigenfunctions for $\el_T$ in $\cX$. In particular $\el_T>0$ on 
$\cX^\perp$.
\end{cor}

To study the spectral gap, we need to understand the action of $\el_K$ on 
products of Hermite polynomials $H_{\ba}$. We first state and prove the 
following lemma which helps us restrict our investigation to even polynomials. 

\begin{lem} \hspace*{\fill}
\begin{itemize}
\item[$\bullet$] Any eigenfunction of $\mu \el_T+\lambda \el_K$ is either even or odd in each variable $v_i$.
\item[$\bullet$] If $E$ is an eigenvalue of $\mu \el_T+\lambda \el_K$, with an 
eigenfunction that is odd in some $v_i$, we have that $E \geq 2\lambda+\mu$. 
\end{itemize}
\end{lem}
\begin{proof}
The first part can be seen by noting that the operator $\mu \el_T+\lambda \el_K$ 
commutes with the reflection operator $S_j [h](\V):=h(...,-v_j,...)$.
For the second part, say $(\mu \el_T+\lambda \el_K)h=Eh$, with $S_1[h]=-h$. Then $T_1[h]=0$. In addition, for any $i \neq 1$, 
\begin{align*}
\dashint{d\theta h(\V_{i,1}(\theta))} &= \dashint{d\theta h(v_i \cos{\theta}+v_1 \sin{\theta},...,-v_i \sin{\theta}+v_1 \cos{\theta},...)} \crcr
&= \dashint{d\theta h(\sqrt{v_i^2+v_1^2}\cos{(\varphi-\theta)},...,\sqrt{v_i^2+v_1^2} \sin{(\varphi-\theta)},...)} \crcr
&= \dashint{d\theta h(\sqrt{v_i^2+v_1^2}\cos{\theta},...,\sqrt{v_i^2+v_1^2} \sin{\theta},...)} \crcr
&= \dashint{d\theta h(-\sqrt{v_i^2+v_1^2}\cos{\theta},...,\sqrt{v_i^2+v_1^2} 
\sin{\theta},...)} \;\;\;\;\;\;\;\; 
\text{(taking $\theta \to \pi-\theta$)}\crcr
&=0 \;\;\;\;\;\;\;\; \text{(using that $S_1[h]=-h$)} \ .
\end{align*}
Thus, 
\[
\lambda N h - \lambda \frac{N}{\binom{N}{2}} \sum_{i<j, i,j \neq 
1}{\dashint{d\theta h(\V_{i,j}(\theta))}} + 
N \mu h - \mu \sum_{i \neq 1}{T_i[h]} = Eh\]
or
\[(\lambda N + \mu N - E) \leq \lambda \frac{N}{\binom{N}{2}} \binom{N-1}{2} + \mu(N-1) \ , \]
which proves the claim. 
\end{proof}

We will thus restrict our attention to the space of functions that are even in 
all variables $v_i$ and show that the eigenfunction for $\Delta_N$ and 
$\Delta_N^{(2)}$ lie in this space. To this end we define
\[
L_{2l}={\rm span}\Bigl\{H_{2\ba}\,\Big|\, \sum_{i=1}^N 
2\a_i=2l\Bigr\}.
\]
Moreover we set 
\[
 |\ba|=\sum_{i=1}^N\a_i
\]
and
\[
\Xi:=\{\ba: \sum_{i<j}{\alpha_i \alpha_j} \neq 0\} \text{, that is the set of 
} \ba \text{ in which at least two entries are non-zero.}
\]
\begin{lem} \label{L:mon}
In each $L_{2l}$ the eigenvalues of $\el_T$ are given by 
$\sigma_{2\ba}=\sum_j{(1-s_{2\alpha_j})}$, where $|\ba|=l$. It follows 
that
\begin{itemize}
\item[$\bullet$] The smallest eigenvalue in each $L_{2l}$ is $1-s_{2l}$ and the 
corresponding eigenfunctions are precisely linear combinations of $H_{2\ba}(\V)$ 
with $\ba=(0,\ldots,l,\ldots,0)$. 

\item[$\bullet$] $\displaystyle\min_{\ba \in \Xi}{\sigma_{2\ba}}=1$. Morover, 
the minimum is reached when two of the $\alpha_i$'s are 1 and the rest are $0$.
\end{itemize}
\end{lem}

\begin{proof}
To prove the first statement, we start by observing that the function 
$J(x):=\displaystyle\dashint_0^{2\pi}{\cos^{2x}{\theta} d\theta}$ is strictly 
convex in $x$. Consider $\ba$ such that $|\alpha|=l$. We need to show that 
\[\sum{J(\alpha_i)} \leq J(l)+(N-1)J(0)\] and that equality is attained if and 
only if $\ba=(0,...l,...0)$.
By convexity, we have that

\begin{align*}
J(\alpha_i) &= J \biggl(\frac{\alpha_i}{l}l+ \sum_{j \neq i}{\frac{\alpha_j}{l}}0 \biggr) \crcr
 &\leq \frac{\alpha_i}{l}J(l)+\sum_{j \neq i}{\frac{\alpha_j}{l}} J(0) \ . 
\end{align*}
Summing the above over $i$, we get the result.

The second claim follows from the monotonicity of the $s_{2\alpha}$ and the fact that $s_2 = \frac{1}{2}$.
\end{proof}

\begin{proof}[Proof of Proposition~\ref{P:1gap}]
By Corollary~\ref{C:elT} and Lemma~\ref{L:mon}, we have that $\el_T \geq 1/2$ 
and thus $\el \geq 1/2$ on $\cX^\perp$. On the other hand, $\el 
[\sum{H_2(v_i)}]=\el_T [\sum{H_2(v_i)}]=\frac{\mu}{2} (\sum{H_2(v_i)})$ since 
$\sum{H_2(v_i)}$, being a radial function is annihilated by the Kac part. Thus, 
$\Delta_N=\mu/2$ and $h_{\Delta_N}=\sum_{i=1}^N{H_2(v_i)} \in L_2$.
\end{proof} 
 
To compute $\Delta_N^{(2)}$ we need to better understand the action of $\el_K$ 
on the $L_{2l}$. This is done in the following Lemma, which is actually a 
generalization of Lemma~\ref{L:herm}.
  
\begin{lem} \label{L:hom}
Let $A$ be a self-adjoint operator on $L^2(\R^N,\gamma(\V)d\V)$ that preserves 
the space $P_{2l}$, of homogeneous even polynomials in $v_1,...,v_N$ of degree $2l$. 
If
\[
A(v_1^{2\alpha_1}...v_N^{2\alpha_N})=\sum_{|\underline{\beta}|=|\ba|} c_{ 
\underline{\beta}}v_1^{2\beta_1}...v_N^{2\beta_N} \ ,
\] 
we get
\[
A(H_{2\alpha_1}(v_1)...H_{2\alpha_N}(v_N))=\sum_{|\underline{\beta}|=|\ba|}
{c_{\underline{\beta}}H_{2\beta_1}(v_1)...H_{2\beta_N}(v_N)} \ .
\]
\end{lem}
\begin{proof}
First, we observe that $A(L_{2l}) \subset L_{2l}$. Indeed, if $f \in L_{2m}$ 
and $g \in L_{2l}$ with $m < l$, we have $\langle Ag,f \rangle=\langle g,Af \rangle=0$ because $Af$ 
contains only monomials of degree at most $2m$. This means that 
\[
A(H_{2\alpha_1}(v_1)...H_{2\alpha_N}(v_N))=\sum_{|\underline{\beta}|=|\ba|}
{k_{\underline{\beta}}H_{2\beta_1}(v_1)...H_{2\beta_N}(v_N)}
\]
and because
\[
A(v_1^{2\alpha_1}...v_N^{2\alpha_N})=\sum_{|\underline{\beta}|=|\ba|}
{c_{\underline{\beta}}v_1^{2\beta_1}...v_N^{2\beta_N}} \ ,
\]
we get that $c_{\underline{\beta}}=k_{\underline{\beta}}$ for any 
$\underline{\beta}$ by equating the coefficients of the term of maximal degree 
$v_1^{2\beta_1}...v_N^{2\beta_N}$. 
\end{proof}
\hspace*{\fill}
\begin{xrems} \hspace*{\fill}
\begin{itemize}
\item[$\bullet$]Since $\el_K$ preserves the spaces $P_{2l}$, the above Lemma applies to it. 
Thus, the action of $\el_K$ on products of Hermite polynomials $H_{2n}(v_i)$ can 
be deduced from its action on products of mononomials $v_i^{2n}$, and the latter 
turns out to be simpler.  
\item[$\bullet$]Note that $L_{2l}$ is invariant under $\el_K$ and thus is invariant 
under $\el$. 
\end{itemize}
\end{xrems}

In preparation for the proof of Theorem~\ref{T:2gap} we note that 
Theorem~\ref{T:CCL} implies that
\[ 
\langle h, \el_K h \rangle \geq \langle h, \Lambda_N (I-B) h \rangle 
\]
where $B$ is the orthogonal projection on radial functions, that is
\[
 B[h](\V)=\int_{S^{N-1}(|\V|)}h({\bf w} )d\sigma ({\bf w}) \ .
\]
where $S^{N-1}(r)$ is the sphere of radius $r$ in $\R^N$ with normalized surface measure  
$d\sigma(\V)$. Setting $\el_R:= 
\Lambda_N (I-B)$ we have
\[ 
\langle h,\el h \rangle \geq \langle h, (\mu\el_T + \lambda\el_R)h \rangle 
\]
so that
\begin{equation} 
\Delta_N^{(2)} \geq \text{ inf} \{ \langle h,(\mu\el_T + \lambda\el_R)h \rangle : ||h||=1, h 
\perp L_0,L_2\} \ ,
\end{equation}
where we have replaced the operator $\el_K$ with the much simpler projection 
$\el_R$. Note, the same reasoning as before shows that the space $L_{2l}$ is 
invariant under $\el_R$. For later use we define
\[
\Gamma(\ba) =\int_{S^{N-1}(1)}v_1^{2\alpha_1}...v_N^{2\alpha_N} d\sigma_1(\V).
\]

\begin{thm} \label{T:simple} The smallest eigenvalue $a_l$ of the operator
\[
\el_S : = \mu\el_T + \lambda\el_R
\]
restricted to the space $L_{2l}$ satisfies the estimates
\[
a_l \ge  x_l \ ,
\]
where $x_l $ is the smaller of the two solutions of the equation
\begin{equation} \label{quadratic}
x^2 - \left(\lambda\Lambda_N + (2-s_{2l})  \mu\right)x + (1-s_{2l})  \mu^2 + 
\lambda\Lambda_N \mu = \lambda \Lambda_N \mu s_{2l} N \Gamma(l,0,...0) \ .
\end{equation}
\end{thm}

\begin{proof}
The equation for the eigenvalue $x$ of $\mu \el_T + \lambda \el_R$ gives
\[
\mu \sum{T_j}h+ \lambda\Lambda_N Bh = (N \mu + \lambda\Lambda_N - x) h \ .
\]
 Observe that if $|\ba|=l$, $B[v_1^{2\alpha_1}\cdots v_N^{2\alpha_N}]$ is an 
homogeneous radial polynomial of degree $2l$ so that we have
\[
 B[v_1^{2\alpha_1}\cdots v_N^{2\alpha_N}](r)=\Gamma(\ba)r^{2l}=
 \Gamma(\ba) \displaystyle\sum_{|\underline{\beta}|=l} 
\frac{l!}{\beta_1!...\beta_N!} v_1^{2\beta_1}...v_N^{2\beta_N} \ ,
\]
in particular
\begin{equation} \label{E:rad2}
  \sum_{|\ba|=l}\frac{l!}{\alpha_1!...\alpha_N!}\Gamma(\ba) = 1.
\end{equation}

Writing a generic function $f$ in $L_{2l}$ as 
\[
f =\sum_{|\ba|=l}c_{\ba}H_{2\ba}
\]
the eigenvalue equation becomes:
 \begin{equation} 
\mu\sum_{|\ba|=l}\sum_j s_{2\alpha_j}c_{\ba}H_{2\ba}+\lambda\Lambda_N 
\biggl[\sum_{|\ba|=l}c_{\ba}\Gamma(\ba)\biggr] 
\sum_{|\ba|=l}\frac{l!}{\alpha_1!...\alpha_N!}H_{2\ba}  
= (N \mu + \lambda\Lambda_N - x) 
\sum_{|\ba|=l}c_{\ba}H_{2\ba} \ ,
\end{equation}
where we have used that the projection $B$ satisfies the hypothesis of Lemma~\ref{L:hom}. 

Thus for every  $\ba$ 
\begin{equation} \label{E:c}
\left(\mu \sigma_{2\ba} + \lambda \Lambda_N - x\right) 
c_{\ba} = K \lambda \Lambda_N \frac{l!}{\alpha_1!...\alpha_N!} \ ,
\end{equation}
where we set $\sum_{|\ba|=l}c_{\ba}\Gamma(\ba)=K$. Consider first the case $K 
\neq 0$, that is $(x -\lambda \Lambda_N - \mu\sigma_{2\ba}) \neq 0$ for every $\ba$. 
Rearranging, multiplying both sides by $\Gamma(\ba)$, and adding we get
  
\begin{equation} \label{E:eig}
\frac 1{\lambda\Lambda_N} = \sum_{|\ba|=l}\frac{1}{\lambda\Lambda_N + \mu 
\sigma_{2\ba} - x} \Gamma(\ba) 
\frac{l!}{\alpha_1!...\alpha_N!} \ .
\end{equation}

With $x$ moving in from $-\infty$, the first singularity of the right side of eq.~\eqref{E:eig} occurs 
when 
\[
x=\min_{|\ba|=l}\left(\lambda\Lambda_N + \mu 
\sigma_{2\ba}\right)=\lambda\Lambda_N +\mu(1-s_{2l}),
\]
where the last equality follows from Lemma~\ref{L:mon}. The right side of 
eq.~\eqref{E:eig} is a positive increasing function of $x$ until the first 
singularity. Thus, the smallest eigenvalue is less than $\lambda\Lambda_N 
+\mu(1-s_{2l})$. For $0<x < \lambda\Lambda_N +\mu(1-s_{2l})$ we get
\begin{align*}  
\frac{1}{\lambda\Lambda_N} &= \frac{1}{\lambda\Lambda_N + (1-s_{2l})\mu -x} N 
\Gamma(l,0,...0) + \sum_{\substack{|\ba|=l\\\ba \in \Xi}} {\frac{1}{(\lambda\Lambda_N + \mu \sigma_{2\ba} - x)} 
\Gamma(\ba) \frac{l!}{\alpha_1!...\alpha_N!}} \crcr
&\leq\frac{1}{\lambda\Lambda_N+(1 - s_{2l}) \mu-x} N \Gamma(l,0,...0) + 
\frac{1}{\lambda\Lambda_N+\mu-x}\sum_{\substack{|\ba|=l\\\ba \in \Xi}}{\Gamma(\ba) 
\frac{l!}{\alpha_1!...\alpha_N!}}\crcr
&\leq \frac{1}{\lambda\Lambda_N+(1 -s_{2l}) \mu-x} N 
\Gamma(l,0,...0) + \frac{1}{\lambda\Lambda_N+\mu-x}[1-N \Gamma(l,0,...0)] \ . \,\,\,\,\, \text{ (using eq.~\eqref{E:rad2})}
\end{align*}
It is easily seen that the equation  
\begin{equation} \label{fraction}
\frac{1}{\lambda\Lambda_N} = \frac{1}{\lambda\Lambda_N+(1 -s_{2l}) \mu-x} N 
\Gamma(l,0,...0) + \frac{1}{\lambda\Lambda_N+\mu-x}[1-N \Gamma(l,0,...0)]
\end{equation}
and \eqref{quadratic} are equivalent and hence the smallest eigenvalue $a_l \ge x_l$.

Note that necessarily $x_l < \lambda \Lambda_N+\mu(1-s_{2l})$. Thus, if $K=0$, $a_l=\lambda \Lambda_N+\mu 
\sigma_{2\ba}$ for some $\ba$ and hence $ a_l \geq \lambda \Lambda_N+\mu(1-s_{2l}) > x_l$
which proves the theorem.
\end{proof}
    
\begin{proof}[Proof of Theorem~\ref{T:2gap}]
Since symmetric functions are preserved under $\el$, the space of symmetric 
Hermite polynomials in $L_4$ with orthonormal basis 
$\{\sqrt{\frac{2}{N(N-1)}}\sum_{i \neq j}{H_2(v_i)H_2(v_j)}, 
\sqrt{\frac{2}{3N}}\sum{H_4(v_i)} \}$ gives rise to two eigenfunctions. The 
action of $\mu \el_T+\lambda \el_K$ on this space is represented by the 
following matrix
\begin{equation} \label{E:matrix}
\begin{pmatrix}
\mu+\frac{3\lambda}{2(N-1)} & \frac{-\sqrt{3}\lambda}{2\sqrt{N-1}} \\
\frac{-\sqrt{3}\lambda}{2\sqrt{N-1}} & \frac{5\mu}{8}+\frac{\lambda}{2} \\
\end{pmatrix}
\end{equation}
whose characteristic equation is~\eqref{E:2gap} and smallest eigenvalue is thus $a_2$. Hence, we immediately have $\Delta_N^{(2)} \leq a_2$. 

To see the opposite inequality recall that $x_l$ is the smaller of the two 
solutions of the equation \eqref{fraction}. Since for $l \ge 2$,  $s_{2l} \le s_4 
= \frac{3}{8}$ and $\Gamma(l,0,...0)  \leq   \Gamma(2,0,...0) =  
\frac{3}{N(N+2)} $ we get from \eqref{fraction} 
\[
\Delta_N^{(2)} \ge  a_2 \ .
\]
\end{proof} 

The eigenfunction corresponding to the ``second'' gap $\Delta_N^{(2)}$ is given 
by $\sum_{|\ba|=2}c_{\ba}H_{2\ba} \in L_4$, where $c_{\ba}$ are symmetric under 
exchange of indices (see eq.~\eqref{E:c}). In fact, eq.~\eqref{E:c} 
characterizes the symmetric eigenfunctions of $\mu\el_T+\lambda\el_R$ when $K 
\neq 0$ and the non-symmetric eigenfunctions when $K=0$. This means that 
solutions $x$  of eq.~\eqref{E:eig} correspond to symmetric eigenfunctions 
alone. Hence, the unique eigenfunction (by extension, also that of 
$\mu\el_T+\lambda\el_K$) corresponding to $a_2$ is symmetric, which is the 
physically interesting case.

We eventually do get the optimal bound $a_2$ due to the following reason: The 
space of symmetric functions in $L_4$ is spanned by the set $\{\sum_{i \neq 
j}{H_2(v_i)H_2(v_j)},\sum{H_4(v_i)}\}$, which can also be spanned by two 
functions, one of which is radial (of degree 4) and the other  perpendicular to 
the radial one. The latter gives the gap $\Lambda_N$ for $\el_K$. Hence, the 
action of $\el_K$ and $\el_R$ on the space of symmetric functions in $L_4$ is 
precisely the same.

In the limit $N \to \infty$, the off-diagonal elements of the matrix~\eqref{E:matrix} vanish. Therefore, in this limit, the eigenvalues 
$x_2^{\pm}$ tend to $\frac \lambda 2 + \frac 5 8 \mu$ and $\mu$, which 
corresponds to the simultaneous diagonalization of operators $\el_T$ and 
$\el_K$.

\subsection{Approach to Equilibrium in Entropy}\label{sec:2.2}
 
It is well known that 
\[
S(f|\gamma)=0 \Leftrightarrow f=\gamma,
\] 
where $S(f|\gamma)$ is defined as in eq.~\eqref{E:relent}. In this section we 
will prove that $S(f_t|\gamma)$ decays to 0 exponentially as $t\to\infty$, if 
$f_t$ is the solution of the Master equation~\eqref{E:mast}. Indeed Theorem~\ref{T:conen} immediately follows from the following proposition.
  
\begin{prop} \label{P:entr}
Let $f_0$ be a probability density on $\R^N$ with finite relative 
entropy and $f_t$ the solution of the Master equation~\eqref{E:mast} with 
initial condition $f_0$. We have
  
\begin{equation}\label{E:produ}
\frac{dS(f_t|\gamma)}{dt} \leq - \rho S(f_t|\gamma) \ ,
\end{equation}
where 
\[
\rho = \frac{\mu}{2} \ .
\]
\end{prop}

The left-hand side of the inequality~\eqref{E:produ} is the entropy 
production and can be computed as
\[
\frac{dS(f_t|\gamma)}{dt} = \int{\frac{\partial f_t}{\partial 
t}\log{\frac{f_t}{\gamma}}} + \int{\frac{\partial f_t}{\partial t}} = 
-\int(\lambda \eg_K + \mu \eg_T)[f_t] \log\frac{f_t}{\gamma} 
\] 
because $\int f_t=1$. Thus, Proposition~\ref{P:entr} will follow 
if we prove that for any density $f$ with finite relative entropy, we have
\[
 -\int(\lambda \eg_K + \mu \eg_T)[f] \log\frac{f}{\gamma}\leq -\rho 
\int f \log\frac{f}{\gamma} \ .
\]
Since the relative entropy decreases along the Kac flow, i.e. $\int \eg_K [f] 
\log\frac{f}{\gamma} \geq 0$ (see \cite{kac}), it is enough to show that

\begin{equation} \label{E:tent}
-\int \mu \eg_T[f] \log\frac{f}{\gamma}\leq -\rho 
\int f \log\frac{f}{\gamma} \ .
\end{equation}
We first prove the above for the case $N=1$. In this case, $\eg_T=(I-R)$.

\begin{prop} \label{P:N1}
Let $f$ be a probability density on $\mathbb{R}$. Then
\[
\int R[f](v) \log \frac{f(v)}{g(v)} dv = \int dv dw \dashint d\theta f(v^*)g(w^*) \log\frac{f(v)}{g(v)} \leq \frac 1 2 
\int dv f(v) \log\frac{f(v)}{g(v)} \ ,
\]
where $v^*=v\cos{\theta}+w\sin{\theta}$, $w^{*}=-v\sin{\theta}+w\cos{\theta}$ and $g(v)$ is the Gaussian $\frac{1}{\sqrt{2\pi}} e^{-\frac{v^2}{2}}$.
\end{prop}

Calling $f(v)=g(v)G(v)$, we need to prove that
\begin{equation} \label{E:tent1}
\int dv\, g(v) T[G](v) \log G(v) \leq \frac{1}{2} \int dv\, g(v) G(v) \log 
G(v) \ ,
\end{equation}
where 
\[
T[G](v):= \int dw g(w)\dashint_0^{2\pi} d\theta G(v\cos\theta+w\sin\theta) \ .
\]

The idea will be to show the above inequality by proving that

\begin{equation} \label{E:tent1eq}
\int dv\, g(v) T[G](v) \log T[G](v) \leq \frac{1}{2} \int dv\, g(v) G(v) \log 
G(v) \ .
\end{equation}
We will be invoking the following well-known property of the Ornstein-Uhlenbeck process, see \cite{Stam,Gross1,BE,Gross2,Toscani}.

\begin{thm} \label{T:ou}
Let $P_s$ be the semigroup generated by the 1-dimensional Ornstein-Uhlenbeck 
process, that is, $U_s=P_s[U_0]$ is the solution of the Fokker-Planck equation
\[
\frac{\partial U_s(v)}{\partial s} = U_s''(v) - vU_s'(v)
\]
with initial condition $U_0$. For every density $G$ we have
\[
\int g(v)dv \, P_s[G](v) \, \log(P_s [G](v)) \leq e^{-2s} \int g(v)dv \, G(v) \, \log G(v) \ .
\]
\end{thm}
 
\begin{xrem}
The semigroup, which can be represented explicitly as 
\begin{equation} \label{E:OUsemigp}
P_s[G](v)=\int{dw g(w) G(e^{-s}v+\sqrt{1-e^{-2s}}w)} \ ,
\end{equation}
is self-adjoint in $L^2(\R,g(v) dv)$.
\end{xrem}

We are now ready to prove Proposition~\ref{P:N1}.

\begin{proof}[Proof of Proposition~\ref{P:N1}]

To connect the Ornstein-Uhlenbeck process $P_s$ with the operator $T$ we set
\[
\overline{T} 
[G](v):=\int{dw g(w)\dashint_0^{\pi/2}{d\theta 
G(v\cos{\theta}+w\sin{\theta})}} = \frac{2}{\pi} \int_0^{\infty} ds 
\frac{e^{-s}}{\sqrt{1-e^{-2s}}} P_s[G](v) \ ,
\]
where we use eq.~\eqref{E:OUsemigp} and the change of variables $cos(\theta)=e^{-s}$. It follows that
\begin{align*}
\int{dv \, g(v) \overline{T}[G] \log{\overline{T}[G]}} &= \int dv \, g(v) \left(\frac{2}{\pi} \int_0^{\infty} ds 
\frac{e^{-s}}{\sqrt{1-e^{-2s}}} P_{s}[G]\right) \log{\left(\frac{2}{\pi} \int_0^{\infty} ds' 
\frac{e^{-s'}}{\sqrt{1-e^{-2s'}}} P_{s'}[G]\right)} \crcr
&\leq \int dv \, g(v) \left(\frac{2}{\pi} \int_0^{\infty} ds 
\frac{e^{-s}}{\sqrt{1-e^{-2s}}} P_{s}[G] \log{P_s[G]} \right)  \;\;\;\;\;\;\;\; \text{(using convexity of $x\log{x}$)}\crcr 
&\leq \frac{2}{\pi} \int_0^{\infty} ds \frac{e^{-s}}{\sqrt{1-e^{-2s}}} e^{-2s} \int dv \, g(v) G \log{G} \;\;\;\;\;\;\;\; \text{(using Theorem~\ref{T:ou})} \crcr
&= \frac{1}{2} \int dv \, g(v) G\log{G} \ .
\end{align*}

The next step is to prove the corresponding result for the operator $T$. Let $G=G_e+G_o$ where $G_e$ is even, i.e. $G_e(v)=G_e(-v)$, and $G_o$ is 
odd, i.e. $G_o(-v)=-G_o(v)$.  Observe that $T[G]$ is even, $T[G_o]=0$ and $\overline{T}[G_e]=T[G_e]$. While the first two identities follow directly from  the definitions, the last one also uses the fact that $\displaystyle\int{dwg(w)\int_{\frac{\pi}{2}}^{\pi}{d\theta G_e(v\cos{\theta}+w\sin{\theta})}} = \displaystyle\int{dwg(w)\int_{0}^{\frac{\pi}{2}}{d\theta G_e(-v\cos{\theta}-w\sin{\theta})}}$ under the change of variables $\theta \rightarrow \pi-\theta$ and $w \rightarrow -w$. Thus,
 
\begin{align*}
\int dv \, g(v) T[G](v)\log T[G](v) =& \int dv \, g(v) T[G_e](v) \log T[G_e](v) \crcr
=& \int dv \, g(v) \overline{T}[G_e](v) \log \overline{T}[G_e](v)  \crcr
\leq & \frac{1}{2} \int dv \, g(v) G_e(v) \log G_e(v)  \crcr
\leq & \frac{1}{2} \int dv \, g(v) G(v) \log G(v) \ ,
\end{align*}
where, in the last inequality, we have used that $G_e(v)=(G(v)+G(-v))/2$ and Jensen's inequality. Now that we have established~\eqref{E:tent1eq}, we proceed to derive~\eqref{E:tent1} from it as follows:

\begin{align*}
\int dv \, g(v) e^{(T-I)t}G \log (e^{(T-I)t}G) &\leq e^{-t} 
\displaystyle\sum_{k=0}^{\infty} \frac{t^k}{k!} \int dv \, g(v) T^k[G](v) \log T^k[G](v) 
\;\;\;\;\;\;\;\; \text{(by convexity)} \crcr
&\leq e^{-t} \displaystyle\sum_{k=0}^{\infty} \left(\frac{t}{2}\right)^k 
\frac{1}{k!} \int dv \, g(v) G(v) \log G(v) \;\;\;\;\;\; \text{(by the previous 
computation)} \crcr
&= e^{-\frac{t}{2}} \int dv \, g(v) G(v) \log G(v) \ .
\end{align*}
From the first order terms in the Taylor expansion about $t=0$, we get \eqref{E:tent1}.
\end{proof}

The following lemma will help extend the result to $N>1$.
  
\begin{lem} \label{L:jen}
Let $f(\V)$ be a probability density on $\R^N$ and let its marginal 
over the $j^{\text{th}}$ variable be denoted by $f_{j} (\hat{\V}_j)= 
\int{f(\V)dv_j}$, where $\hat{\V}_j=(v_1,\ldots,v_{j-1},v_{j+1},\ldots, v_N)$. 
Then we have 
\[ 
\sum_{j=1}^N \int f_{j} \log f_{j} d\hat{\V}_j \leq (N-1) \int f \log f 
d\V \ .
\]
\end{lem}

\begin{proof}
We first observe that from the Loomis-Whitney inequality \cite{LW}, that is $\displaystyle\int_{\R^N}{F_1(\hat\V_1)...F_N(\hat\V_N)} \leq ||F_1||_{L^{N-1}}...||F_N||_{L^{N-1}}$ for $F_j \in L^{N-1}(\R^{N-1})$, it follows 
that
\begin{equation} \label{E:young}
Z:=\int{\prod_{j=1}^{N} f_j^{\frac{1}{N-1}} d\V} \leq 1 \ .
\end{equation}
  
Thus we have
\begin{align*}
\int f \log\left[\frac{f}{\prod f_j^{\frac{1}{N-1}}}\right]d\V 
=& Z \int \frac{f}{\prod f_j^{\frac{1}{N-1}}} 
\log \left[\frac{f}{\prod f_j^{\frac{1}{N-1}}}\right] 
\frac{\prod f_j^{\frac{1}{N-1}}}{Z} d\V 
\geq Z\left[\int \frac{f}{Z}d\V \right] \log\left[\int 
\frac{f}{Z}d\V\right] 
= -\log{Z} \ ,
\end{align*}
where we have used Jensen's inequality and the convexity of $x\log(x)$. The Lemma follows easily from the 
above inequality and~\eqref{E:young}.
\end{proof}
   
\begin{proof} [Proof of Proposition~\ref{P:entr}]

We first observe that
\begin{align*}
-\int \eg_T[f] \log\frac{f}{\gamma} =& \sum_j \int d\V \int dw 
\dashint d\theta f(\V_j(\theta,w))g(w_j^*(\theta)) \log\frac{f(\V)}{\gamma(\V)} - 
NS(f|\gamma)\crcr
=& \sum_j \int d\hat{\V}_j f_j(\hat{\V}_j) \int dv_j dw 
\dashint d\theta 
\frac{f(\V_j(\theta,w))}{f_j(\hat{\V}_j)}g(w_j^*(\theta))\log\left(\frac{f(\V)}{
f_j (\hat{\V}_j)
g(v_j)}\right)   \crcr
&\quad+\sum_j \int d\V dw \dashint d\theta f(\V_j(\theta,w))g(w_j^*(\theta))
\log\frac{f_j(\hat{\V}_j)}{\gamma_j(\hat{\V}_j)} - NS(f|\gamma) \ ,
\end{align*}
where $f_j(\hat{\V}_j):=\int{dv_j f(\V)}$ and $\gamma_j(\hat{\V}_j):=\int{dv_j 
\gamma(\V)}$, as in Lemma~\ref{L:jen}. In the first term of the last line, we can undo the rotation by $\theta$ by noting that $f_j$ and $\gamma_j$ are independent of $v_j$. Applying Proposition~\ref{P:N1} 
to $\frac{f(\V)}{f_j(\hat\V_j)}$ (in the second line above) as a function of $v_j$ alone we get:

\begin{align*}
-\int \eg_T[f] \log\frac{f}{\gamma} & \leq \frac{1}{2} \sum_j \int d\V f(\V) \log \left(\frac{f(\V)}{f_j(\hat\V_j) g(v_j)}\right) + \sum_j \int d\V f(\V) \log \frac{f_j(\hat\V_j)}{\gamma_j(\hat\V_j)} -NS(f|\gamma) \crcr
&= \frac{1}{2} \sum_j \int d\V f(\V) \log \left(\frac{f(\V)}{f_j(\hat\V_j) g(v_j)}\right) + \frac{1}{2}\sum_j \int d\V f(\V) \log \frac{f_j(\hat\V_j)}{\gamma_j(\hat\V_j)} \crcr
& \quad \quad + \frac{1}{2}\sum_j \int d\V f(\V) \log \frac{f_j(\hat\V_j)}{\gamma_j(\hat\V_j)} -NS(f|\gamma) \crcr
  &= -\frac{1}{2}NS(f|\gamma) + \frac{1}{2} \sum_j \int d\V f(\V)\log{f_j}(\hat\V_j) -\frac{1}{2} \sum_j \int d\V f(\V)\log{\gamma_j(\hat\V_j)} \ . \crcr
\end{align*}

Using Lemma~\ref{L:jen} for the second term, and that 
$\gamma_j(\hat\V_j)=\prod_{i\not=j}g(v_i)$ so that $\sum_j \int d\V f(\V) 
\log{\gamma_j(\hat\V_j)} = (N-1) \int{f \log \gamma}$ we get:  
\[
-\int \eg_T[f] \log\frac{f}{\gamma} \leq -\frac{1}{2}S(f|\gamma) 
\]
and this proves~\eqref{E:tent}.
\end{proof}

\begin{xrems}  \hspace*{\fill}
\begin{itemize}
\item[$\bullet$] Proposition~\ref{P:entr} yields a lower bound on the spectral gap $\Delta_N$ as follows: given a function $f$ of the form
\[
f = \gamma (1+ \epsilon h) 
\]
with $\int{h \gamma}=0$ and $\epsilon$ small, one can write
 \[ 
\epsilon \int \gamma \frac{\partial h}{\partial t}\left(\epsilon h - 
\frac{\epsilon^2h^2}{2}...\right) \leq 
- \rho \int\gamma(1+\epsilon h)\left(\epsilon h-\frac{\epsilon^2h^2}{2}...\right) \ ,
\]
where $\rho=\mu/2$. That is,
\[ 
\int\gamma h \frac{\partial h}{\partial t} \leq - \rho \int\frac{\gamma 
h^2}{2} \ .
\]
Thus in $L^2(\R^N,\gamma(\V)d\V)$ we get
\[
 \frac{d}{dt}\Vert h\Vert \leq - \frac{\rho}{2} \Vert h\Vert \ .
\]
Observe that this is very similar to the result one get from Proposition~\ref{P:1gap} but $\rho<\mu$. One may wonder whether $\rho$ is the optimal 
estimate for the decay rate of the relative entropy.

\item[$\bullet$] In contrast to the Kac model, the presence of the thermostat guarantees that the 
rate of convergence is strictly positive uniformly in $N$. It is 
fundamental in the above analysis that the thermostat acts on all particles. The presence of the Kac part gives no contribution to the above estimate of the exponential decay rate.
\end{itemize}
\end{xrems}

\section{Propagation of Chaos} \label{sec:3}

We finally turn our attention to the effective Boltzmann equation that emerges 
in the limit for $N\to\infty$. The fundamental step to this end is to show 
that the dynamics defined by eq.~\eqref{E:mast} propagates chaos, which is done in this section.

\begin{thm} \label{T:chaos}
 Let $f^{(N)}(\V,0)$ be a chaotic sequence of initial densities. Then its 
evolution under the master equation~\eqref{E:mast},  $f^{(N)}(\V,t)$, is a 
chaotic sequence for any fixed $t$. That is, if
\[
\lim_{N \to 
\infty}{\int_{\R^N}{\varphi_1(v_1)...\varphi_k(v_k)f^{(N)}(\V,0)}} = 
\prod_{j=1}^k{\lim_{N \to \infty}{\int_{\R}{\varphi_j(v_j)f^{(N)}(\V,0)}}}
\]
for any $k \in \mathbb{N}$ and any $\phi_1(v_1),...\phi_k(v_k)$ 
bounded and continuous, then for any $t$:
\[
\lim_{N \to 
\infty}{\int_{\R^N}{\varphi_1(v_1)...\varphi_k(v_k)f^{(N)}(\V,t)}} = 
\prod_{j=1}^k{\lim_{N \to \infty}{\int_{\R}{\varphi_j(v_j)f^{(N)}(\V,t)}}}
\]
for any $k \in \mathbb{N}$ and any $\varphi_1(v_1),...\varphi_k(v_k)$ 
bounded and continuous.
\end{thm}

The proof follows closely the McKean \cite{McK} algebraic 
version of Kac \cite{kac}, for the Kac operator. The idea is to 
write $f(\V,t)=e^{-(\lambda \eg_K+\mu \eg_T)t}f(\V,0)$, expand the exponential in series 
of $t$, and use the chaotic property of the initial sequence. The key observation 
is that $\eg_T$ is a derivation already for finite $N$ (in the sense of Lemma~\ref{L:der}). Two main ingredients 
are needed:

\begin{lem} \label{L:ac}
The series $\sum_{l=0}^{\infty}{\frac{t^l}{l!} 
\int{\varphi_1(v_1)...\varphi_k(v_k) (\lambda \eg_K+\mu \eg_T)^l f(\V,0)}}$ 
converges absolutely if $t<\frac{1}{4\lambda+\mu}$.
\end{lem}

\begin{proof}
To prove the lemma, it is enough to show that:
\begin{equation} \label{E:ac}
||(\lambda \eg_K + \mu \eg_T^*)^l \phi||_{\infty} \leq 
(4\lambda+2\mu)^l 
m(m+1)...(m+l-1) ||\phi||_{\infty}
\end{equation}
and then follow the proof in \cite{McK}. The above statement follows from a simple 
induction starting from 

\[
|(\lambda \eg_K + \mu \eg_T^*)\phi(v_1,...,v_m)| \leq |\lambda \eg_K \phi| + 
|\mu \eg_T^* \phi| \leq (4\lambda+2\mu) m ||\phi||_{\infty} \ .
\]

\end{proof}
Calling
\[
\Gamma_K \phi := 2 \sum_{i \leq m}\dashint d\theta (\phi(...,v_i 
\cos{\theta}+v_{m+1}\sin{\theta},...) - \phi) \ ,
\]
one can prove, as in \cite{McK}, that if $\varphi_1(v_1),...,\varphi_k(v_k)$ are 
bounded and continuous then:
\[
\lim_{N \to \infty} \int (\lambda \eg_K+\mu \eg_T^*)^l [\varphi_1...\varphi_k] 
f^{(N)}(\V,0) = \lim_{N \to \infty} \int{(\lambda \Gamma_K+\mu \eg_T^*)^l 
[\varphi_1...\varphi_k] f^{(N)}(\V,0)} \ .
\]

The main ingredient to re-sum the power series expansion and obtain the 
Boltzmann equation is the following ``algebraic'' Lemma.

\begin{lem} \label{L:der}
If $(\phi \otimes \psi )(v_1,...,v_{m+k}):= \phi(v_1,...,v_m) 
\psi(v_{m+1},...,v_{m+k})$, then
\[
(\Gamma_K+\eg_T^*)[\phi \otimes \psi] = (\Gamma_K+\eg_T^*)[\phi] 
\otimes \psi + \phi \otimes (\Gamma_K+\eg_T^*)[\psi] \ .
\]
\end{lem}

It is now possible to prove Theorem~\ref{T:Bolt} by following
the  proof in \cite{McK} step-by-step.

\section{Conclusion and Future Work} \label{sec:4}

We hope to have convinced the reader that master equations of Kac type are 
reasonable models for large particle systems interacting with thermal 
reservoirs. The main advantage is that physically relevant quantities such the 
first and second gap can be computed quite easily and the entropic convergence 
to equilibrium can be established in a quantitative fashion as well. Moreover, 
since propagation of chaos holds, contact is made with a Boltzmann type equation 
in one dimension.

There are a number of directions for future research. The generalization to 
three-dimensional momentum-conserving collisions \cite{CGL}, while more complicated, should 
not pose any new real difficulties. There are other, more severe, assumptions 
made in our model that one ought to address.

For example, it is the very nature of a reservoir that it is not influenced by 
the interaction with the other $N$ particles.  A more realistic situation would 
be to consider the reservoir as finite but large. More precisely, consider an 
initial state of the form $\gamma F$ where $\gamma$ is a Gaussian in $M$ 
variables with inverse temperature $\beta$ and $F$ a function of $N$ variables 
with kinetic energy $eN$. Thus, $M$ particles are in thermal equilibrium and $N$ 
particles are not in equilibrium but with finite energy per particle. Now we let 
this state evolve under the Kac evolution. Clearly, this state will evolve to a 
radial function in $N+M$ variables as $t \to \infty$. One would expect that this 
function is close to a Gaussian with a temperature $(\frac{M}{\beta} + 
2eN)/(M+N)$. Assuming that $M>>N$, is it true that the entropy has a rate of 
decay that is uniform in $N$? Note that the problem is not to get an estimate on 
the entropy production of the initial state. The infinitesimal time evolution is 
precisely the weak thermostat treated in our paper. The real issue is to 
quantify the entropy production at later times for which the state is no longer 
of this simple form.

Another important issue is how to understand non-equilibrium steady states in a 
wider sense. We believe that the Kac approach to kinetic theory could shed some 
light on this very difficult problem. The fact that the particles interact with 
a single heat reservoir leads to a non-self-adjoint operator that can be brought 
into a self-adjoint form using a ground state transformation. It is easy to 
write down the master equation for a system of particles that interact with, 
say, two reservoirs at different temperatures. However, the generator cannot be 
brought into a self-adjoint form anymore and the equilibrium cannot be found 
though an optimization procedure, or at least not an obvious one. In particular 
the equilibrium is not a simple function. How, then, can one measure the 
approach to steady-state for such systems? Note that the steady-state is, from 
the point of physics, not an equilibrium state, since it mediates an energy 
transport from a reservoir of higher temperature to one of lower temperature.  
The solution of this problem would be a small step towards understanding 
non-equilibrium steady states.

\section*{Acknowledgements}
The authors would like to thank Eric Carlen, Joel Lebowitz and Hagop Tossounian 
for many enlightening discussions. We thank Hagop Tossounian for helping in 
the proof of Proposition~\ref{P:N1}. We also thank the referee 
for various helpful comments and the extremely detailed report.

\bibliographystyle{habbrv}
\bibliography{nonequi}

\end{document}